\documentclass[conference]{IEEEtran}
\usepackage[T1]{fontenc}
\usepackage[latin9]{inputenc}
\usepackage[english]{babel}
\usepackage{verbatim}
\usepackage{units}
\usepackage{amsthm}
\usepackage[cmex10]{amsmath}
\usepackage{amssymb}
\usepackage{psfrag}
\usepackage{xspace}
\usepackage{amsfonts}
\usepackage{amsthm}
\usepackage{float}
\usepackage{mathbbol}
\usepackage{xargs}
\usepackage{ulem}
\usepackage{oldprsn}
\usepackage{cite}
\usepackage{array}

\usepackage[unicode=true, bookmarks=true,bookmarksnumbered=false,bookmarksopen=false, breaklinks=false,pdfborder={0 0 1},backref=section,colorlinks=false] {hyperref} \hypersetup{pdftitle={An example PDF file from LaTeX}, pdfauthor={Diego Santa Cruz}, pdfsubject={From LaTeX to PDF}, pdfkeywords={PDF, LaTeX, hyperlinks, hyperref}, dvips,bookmarks}
\usepackage{breakurl}

\makeatletter
\theoremstyle{plain}
\newtheorem{thm}{\protect\theoremname}
\theoremstyle{plain}
\newtheorem{lem}[thm]{\protect\lemmaname}
\theoremstyle{remark}
\newtheorem{rem}[thm]{\protect\remarkname}

\newcommand{\mathsym}[1]{{}}

\newcommand{\iid}{\emph{i.i.d.}\xspace}
\newcommand{\prb}[1]{P(#1)}     
\newcommand{\rxRV}{r^\rv{x}}      

\newcommand{\LV}[1]{}
\newcommand{\upto}{,\dots,}
\newcommand{\ignore}[1]{}
\newcommand{\longpaper}[1]{}

\floatstyle{boxed}
\makeatother

\providecommand{\lemmaname}{Lemma}
\providecommand{\remarkname}{Remark}
\providecommand{\theoremname}{Theorem}

\begin{document}

\title{Alternating Markov Chains for Distribution Estimation in the Presence of Errors}
\author{%
\IEEEauthorblockN{Farzad Farnoud (Hassanzadeh)}
\IEEEauthorblockA{Department of Electrical and \\Computer Engineering\\UIUC\\ Urabna, IL 61801\\ Contact: www.ifp.uiuc.edu/\textasciitilde hassanz1}
\and
\IEEEauthorblockN{Narayana P. Santhanam}
\IEEEauthorblockA{Department of\\ Electrical Engineering\\University of Hawaii at Manoa\\Honolulu, HI 96822\\Email: nsanthan@hawaii.edu}
\and
\IEEEauthorblockN{Olgica Milenkovic}
\IEEEauthorblockA{Department of Electrical and \\Computer Engineering\\UIUC\\ Urabna, IL 61801\\ Email: milenkov@uiuc.edu}
}
\maketitle
\begin{abstract}
We consider a class of small-sample distribution estimators over noisy channels. Our estimators  are designed for repetition channels, and rely on properties of the runs of the observed sequences. These runs are modeled via a special type of Markov chains, termed alternating Markov chains. We show that alternating chains have redundancy that scales sub-linearly with the lengths of the sequences, and describe how to use a distribution estimator for alternating chains for the purpose of distribution estimation over repetition channels.
\end{abstract}
\global\long\def\qhalf{q_{1/2}}
\global\long\def\qhalfn{q_{1/2}^{n}}
\global\long\def\ptn{\mathbf{p}}
\global\long\def\prb{P}
\global\long\def\rxRV{r^{\mathbf{x}}}
\global\long\def\pkr#1#2#3{\ptn\left(#1;#2,#3\right)}
\global\long\def\ptns#1{\ptn\left(#1\right)}
\global\long\def\pkrr#1#2#3{\ptn\left(#1;#2,\left[#3\right]\right)}
\global\long\def\pkkr#1#2#3{\ptn\left(#1;\left[#2\right],#3\right)}
\global\long\def\pkkrr#1#2#3{\ptn\left(#1;\left[#2\right],\left[#3\right]\right)}
\global\long\def\pkk#1#2{\ptn\left(#1;\left[#2\right]\right)}
\global\long\def\pk#1#2{\ptn\left(#1;#2\right)}

\section{Introduction}

The problem of estimating the distribution of a source with a large alphabet, based on a small number of observations, is of significant interest in molecular biology, neuroscience, physics, statistics, and learning theory. In order to address this problem, throughout the years a number of sophisticated classes of estimators were developed by Good and Turing \cite{good_population_1953,gale_good-turing_1995} and Orlitsky et al \cite{orlitsky_always_2003,orlitsky-convergence2005}, to cite a few. The idea behind these estimators is to use frequencies of symbol frequencies, rather than simple frequency counts standardly used for Maximum Likelihood (ML) estimation.

An additional problem in this estimation setting arises when some of the observations are inaccurate. Since most known distribution estimators are based on frequency counts, errors that change these counts may have a significant bearing on the accuracy of the method. A particularly interesting case is when the counts are changed by consecutive repetitions of some symbols. 

In \cite{small_sample_sticky}, we described a collection of distribution estimators, based on Expectation Maximization (EM), both for channels with known and channels with unknown repetition parameters. The focal point of the study was a class of sequences, termed alternating sequences, generated by a Markov chain of special topology. The goal of this work is twofold. The first goal is to establish a rigorous analytical framework for evaluating the redundancy of alternating Markov chains. The second goal is to describe how to use alternating sequence distribution estimators for distribution estimation over repetition channels. In particular, we exhibit block and sequential estimators for alternating sequences that have vanishing redundancy and provide for accurate estimation in the presence of repetitions.

It is important to observe that although alternating sequences are generated by a Markov chain, their redundancy cannot be accurately computed using the methods developed in~\cite{Dhulipala2006}. This is due to the fact that the bound of~\cite{Dhulipala2006} are too general to give useful redundancy characterizations for special classes of Markov chains. Surprisingly, the special class of alternating Markov sequences has properties
that may be analyzed using tools developed for i.i.d. sequences, with some appropriate modification. Our analysis reveals that alternating Markov chains have sub-linear pattern redundancy in the sequence length $n$, scaling as $c\,\sqrt{n}+\log(n)$, for some constant $c$. This is a counterpart of the examples described in~\cite{Dhulipala2006}, where Markov chains of redundancy of order $n\log n$ were constructed using permutation patterns. 

The paper is structured as follows. In Section~\ref{sec:small-sample-distn-estimation},
we introduce the problem of distribution estimation over noisy channels and the notion of alternating sequences. In Section \ref{sec:Patterns,-Profiles,-and},
we review the basic ideas and terminology behind the proposed estimation method, including the notions of patterns and profiles. Sections \ref{sec:UB} and \ref{sec:lower-bound} are
devoted to deriving upper and lower bounds on the redundancy of
block estimators for alternating sequences, respectively. A sequential estimator for alternating
sequences is presented
in Section \ref{sec:Sequential-Estimators}. Section \ref{sec:recovery} describes how to determine the source distribution, observed through a noisy channel, based on the estimated probabilities of alternating sequences.

\section{Small-Sample Distribution Estimation\label{sec:small-sample-distn-estimation}}

Consider a sample sequence $\overline{x}$ generated by an \iid source $\mathcal{S}$ defined over a large-cardinality alphabet $\mathcal{A}$. Suppose that the source $\mathcal{S}$ has distribution $p_{\mathcal{S}}$.
The estimator observes an erroneous version of $\overline{x}$,
denoted by $\overline{y}$. The errors are modeled as arising from
a channel $\mathcal{C}$ with input $\overline{x}$ and output $\overline{y}$.
What are the ultimate performance limits for estimating the distribution of the source, given that
the length of $\overline{x}$ is small compared to $|\mathcal{A}|$ or comparable to $|\mathcal{A}|$?

Estimating the distribution of a source based on a noise-free short sequence $\overline{x}$ is a challenging task. ML estimators perform poorly in this setting as they typically overestimate probabilities of seen symbols, while they underestimate probabilities of unseen symbols. More appropriate solutions for this scenario are due to Good and Turing \cite{good_population_1953,gale_good-turing_1995} and Orlitsky et al \cite{orlitsky_always_2003,orlitsky-convergence2005}.

There are two approaches one can follow to address an even more difficult family of problems, namely that of small-sample distribution estimation in the presence of errors.

In one scenario, one may try to first denoise the output sequence $\overline{y}$, so as to obtain an estimate $\hat{\overline{x}}$ of $\overline{x}$, and then apply a small-sample distribution estimator to $\hat{\overline{x}}$. Note that $\hat{\overline{x}}$ depends on both $\overline{y}$ and $p_{\mathcal{S}}$, and thus one needs to estimate $p_{\mathcal{S}}$ in order to estimate $\overline{x}$ and vice versa. This ``estimation loop'' may be resolved via the use of iterative methods that alternate in improving estimates for $\overline{x}$ and $p_{\mathcal{S}}$~\cite{small_sample_sticky}. In another scenario, one may try to first estimate the distribution of $\overline{y}$ and then reconstruct the distribution of $\overline{x}$ by ``inversion'' of the noisy channel. Here, we pursue the second line of reasoning.

The focal point of our inversion study is a special class of Markov sequences that arise in the study of distribution estimation over repetition channels. A repetition channel is a channel which outputs several copies of each input symbol. The number of copies is a random variable with a predetermined distribution of known or unknown paramters. 
One important property of repetition channels is that they maintain the identity and order of symbols in the sequence, and only alter the symbols' \emph{runlengths}.  As an example, the sequence $\overline{x}=$`$committee$' passed through a repetition channel may be observed as $\overline{y}=$`$ccommmiitttee$'. The \emph{alternating sequence} of a sequence $\overline{x}$, $V( \overline{x} ) $, is a sequence obtained from $\overline{x}$ by replacing each run of $\overline{x}$ by one single symbol. We refer to $V( \overline{x} ) = V ( \overline{y})$ as the alternating sequence and denote it by $ \overline{v}$. Note that $ \overline{v} $ is a Markov sequence and its corresponding Markov chain is referred to an alternating Markov chain.

Throughout the rest of the paper, we reserve the symbol $N$ to denote the length of the source output $\overline{x}$. We also use $m$ to denote the cardinality of the alphabet $\mathcal{A}$, which may be infinite, and $n$ to denote the length of the alternating sequence $\overline{v}$.

\section{Patterns, Profiles, and Technical Preliminaries \label{sec:Patterns,-Profiles,-and}}

The pattern $\overline{\psi}=\Psi\left(\overline{x}\right)$ of a sequence $\overline{x}$ is obtained by replacing each symbol by its order of appearance in $\overline{x}$. The profile of $\overline{\psi}$ of a pattern is a vector $\Phi(\overline{\psi})=(\varphi_{1},\cdots,\varphi_{n})$, where $\varphi_{i}$ is the number of symbols that appear $i$ times in $\overline{\psi}$. We use the shorthand notation $\Phi(\overline{x})$ to denote $\Phi(\Psi(\overline{x}))$. Notational confusion can be avoided by noting whether the argument of $\Phi(\cdot)$ is a sequence or a pattern. 

For example, the profile of the pattern $\overline{\psi}=1232421$ is $\varphi=(2,1,1,0,0,0,0)$, since 3 and 4 appear once, 1 appears twice, and 2 appears three times. Note that many patterns may have the same profile. 

It is clear that $\overline{\psi}$ is the pattern of an alternating
sequence $\overline{v}$ of length $n$ if and only if $\psi_{i}\neq\psi_{i+1},$
for $1\le i\le n-1$. As will be seen later, a necessary and sufficient
condition for $\overline{\varphi}$ to be the profile of some alternating
sequence is that $\varphi_{i}=0$ for $i>\left\lceil n/2\right\rceil $. 

\begin{rem}
Observe that every profile of patterns of length $n$ corresponds to a partition of the integer $n$. The number of parts of size $i$ is $\varphi_{i}$ and $\sum_{i}^{n}i\varphi_{i}=n$. In the example above, $\overline{\varphi}=(2,1,1,0,0,0,0)$ corresponds to an (unordered) partition of 7 into two parts of size 1, one part of size 2, and one part of size 3, i.e., 7=1+1+2+3. In the correspondence between partitions and profiles, the number of parts of size $\mu$ equals the number of symbols that appear $\mu$ times. The number of partitions of $n $ is denoted by $\ptns n $.
\end{rem}

\ignore{We find the following notation regarding partitions of integers useful in our subsequent derivations. Let the number of partitions of $n$ be denoted by $\ptns n$, the number of partitions of $n$ into exactly $k$ parts be denoted by $\pk nk$, and let the number of partitions of $n$ into $k$ parts with largest part $r$ be denoted by $\pkr nkr$. An argument with a pair of brackets denotes the maximum allowable value. For example, $\pkrr nkr$ will be used to denote the number of partitions of $n$ into exactly $k$ parts with no part larger than $r$ and $\pkkr nnr$ will be used to denote the number of partitions of $n$ with largest part $r$.}

Let $\mathcal{I}^{n}$ be the collection of \iid distributions over
length $n$ sequences. Consider a probability distribution $p$ over $\mathcal{A}$ 
that assigns probability $0< p_{s} \leq 1$
to $s\in\mathcal{A}$. Then, the distribution induced by $p$ over $\mathcal{A}^{n}$ is denoted by $p^{n}$ and is defined in such a way that, for all $ \overline{x} \in \mathcal A^n$,
\[
p^{n}\left(\overline{x}\right):=\prod_{j=1}^{n}p_{x_{j}}.
\]
Note that $p^{n}\in\mathcal{I}^{n}$.

Every distribution $p$ over $\mathcal{A}$ also induces a distribution $p_{\mathcal{I}_{\Psi}^{n}}$,
\[
p_{\mathcal{I}_{\Psi}^{n}}\left(\overline{\psi}\right):=p^{n}\left(\left\{ \overline{x}:\Psi\left(\overline{x}\right)=\overline{\psi}\right\} \right),
\]
over patterns of length $n$. The set of all such induced distributions is denoted by $\mathcal{I}_{\Psi}^{n}$. We simply write $p\left(\overline{x}\right)$
and $p\left(\overline{\psi}\right)$ if dropping the subscripts and
superscripts causes no confusion. 

Furthermore, $p$ induces a probability distribution over alternating
sequences of length $n$, which for $\overline{v}=v_{1}\cdots v_{n}$ takes the form 
\begin{equation}
p_{\mathcal{V}^{n}}\left(\overline{v}\right):=p_{v_{1}}\prod_{j=2}^{n}\frac{p_{v_{j}}}{1-p_{v_{j-1}}}.\label{eqn:p-alt}
\end{equation}
The set of all such induced distributions is denoted by $\mathcal{V}^{n}$.
The induced distribution $p_{\mathcal{V}_{\Psi}^{n}}$ over alternating
patterns of length $n$ and the set $\mathcal{V}_{\Psi}^{n}$ are defined similarly to their unconstrained sequence counterparts.
Note that ${\mathcal{V}^{n}}$ is Markovian.

\subsection{Properties of Alternating Sequences \label{sec:properties}}

The first issue we address is the relationship between the length of the source sequence and the length of the corresponding alternating sequence. 
The following lemma will be useful in our subsequent discussion.

\begin{lem} \label{lem:n-v-N} Assuming  all symbol probabilities are smaller than $\nicefrac{1}{2}$, we have 
\begin{equation}
P[n\le N/2-\sqrt {8 N \ln N}]<\frac 1N,
\end{equation}
so that with high probability, $n\ge N/2 - \sqrt {8 N \ln N}$.
\end{lem}
\begin{proof}
For an \iid sequence $\overline{x}=x_{1}\cdots x_{N}$, let $R\left(\overline{x}\right)$ denote the number of runs of $\overline{x}$. Let $n_{j},0\le j\le N,$
be a martingale sequence defined as
\[
n_{j}=E\left[R\left(\overline{x}\right)|x_{1} \cdots x_{j}\right].
\]
This exposure martingale if obtained by revealing the values of $x_{i},1\le i\le N,$ one by one - i.e., the filtration exposes the values of $x_{i},1\le i\le N,$ sequentially.

Note that $n_{0}=E\left[R\left(\overline{x}\right)\right]$ and $n=n_{N}=R\left(\overline{x}\right)$.
If $\overline{x}_{1}$ and $\overline{x}_{2}$ differ in only one
symbol, then $\left|R\left(\overline{x}_{1}\right)-R\left(\overline{x}_{2}\right)\right|\le2$.
This implies \cite[Theorem 7.4.1]{alon2008probabilistic} that 
\[
\left|n_{j+1}-n_{j}\right|\le2.
\]
 Hence, for all $\lambda>0$ \cite[Theorem 7.4.2]{alon2008probabilistic},
\begin{equation}\label{eq:mrtngl-prb}
P[n\le E[n]-2\lambda\sqrt N]<e^{-\lambda^{2}/2}
\end{equation}
and thus $n$ is concentrated around its mean. To determine the tail behavior of the distribution of $n$, 
we need to find the expected number of runs $E\left[n\right]$
of an \iid sequence $\overline{x}$. Note that $n$ is a simple random variable,
\[
n=\sum_{i=1}^{N}I_{i}
\]
where $I_{i}$ is indicator function of the event that a run starts. More precisely, $I_{i}=1$ iff a run starts at position
$i$ of $\overline{x}$, and $I_{i}=0$, otherwise. 

For $i=1$, we have
$I_{1}=1$ and for $2\le i\le N$, we have 
\[
P\left[I_{i}=1\right]=\sum_{a\in\mathcal{A}}p_{a}\left(1-p_{a}\right)=1-\sum_{a\in\mathcal{A}}p_{a}^{2}.
\]
Hence,
\begin{align*}
E\left[n\right] & =\sum_{i=1}^{N}E\left[I_{i}\right] =1+\left(N-1\right)\left(1-\sum_{a\in\mathcal{A}}p_{a}^{2}\right).
\end{align*}
Assuming all symbol probabilities are smaller than $1/2 $ implies $\left(1-\sum_{a\in\mathcal{A}}p_{a}^{2}\right)\ge1/2$ and thus
$E[n]\ge N/2$. Thus, under this assumption, form \eqref{eq:mrtngl-prb} with $\lambda = \sqrt{2\ln N}$, the lemma follows.
\end{proof}

The second property, stated in the following lemma, concerns patterns with same probability. The lemma that follows provides the means for analyzing alternating sequences using the techniques developed
for i.i.d sequences.
\begin{lem}\label{lem:same-multip}
Let $\overline{\psi}=\psi_{1}\psi_{_{2}}\cdots\psi_{n}$ and $\overline{\psi}'=\psi'_{1}\psi'_{_{2}}\cdots\psi'_{n}$ be two alternating patterns with profile $\overline \varphi$ such that the multiplicity of $\psi_n$ is equal to the multiplicity of $\psi'_n$. For any \iid distribution $p$, we have $p(\overline\psi)=p(\overline\psi')$.
\end{lem}
\begin{proof}
Let the alphabet be
$\mathcal{A}=\left\{ a_{1},a_{2},\cdots\right\} $ and let the probability of $a_{i}$ be $p_{i}$. Assume $\psi_n=a$ and $\psi'_n=b$, with $a,b\in[k]$, where $k$ is the number of elements appearing in $\overline{\psi}$. Let the multiplicity of $i\in[k]$ in $ \overline\psi$ be denoted by $\mu_i$ and the multiplicity of $i$ in $ \overline\psi'$ be denoted by $\mu'_i$. Note that, by assumption, $\mu_a=\mu'_b$.

Furthermore, let $f$ be a bijection
between the set $\left\{ 1,2,\cdots,k\right\} $ and a subset $A\subseteq\mathcal{A}$
of size $k$. This bijection basically determines what symbol of the
alphabet goes to what symbol in the pattern. Then the probability assigned to pattern $\overline{\psi}$
is 
\begin{align*}
p\left(\overline{\psi}\right) & =\sum_{f}p_{f\left(\psi_{1}\right)}\prod_{j=2}^{n}\frac{p_{f\left(\psi_{j}\right)}}{1-p_{f\left(\psi_{j-1}\right)}}\\
&=\sum_{f}
\frac{\left(p_{f(a)}\right)^{\mu_{a}}}{\left(1-p_{f\left(a\right)}\right)^{\mu_a-1}}
\prod_{i\in[k]\backslash\{a\}}\left(\frac{p_{f\left(i\right)}}{1-p_{f\left(i\right)}}\right)^{\mu_{i}}
\end{align*}
where the summation is over all bijections between the set $\left\{ 1,2,\cdots,k\right\} $
and a subset $A\subseteq\mathcal{A}$ of size $k$. There exists a permutation $g$ over $[k]$ with $g(b)=a$ and, $g(j)=i$ for all $j\in[k]\backslash\{b\}$, 
such that $\mu_i = \mu'_j$. Then, by letting $f'(\cdot)=f(g(\cdot))$, we have
\begin{align*}
&\frac{\left(p_{f(a)}\right)^{\mu_{a}}}{\left(1-p_{f\left(a\right)}\right)^{\mu_a-1}}
\prod_{i\in[k]\backslash\{a\}}\left(\frac{p_{f\left(i\right)}}{1-p_{f\left(i\right)}}\right)^{\mu_{i}}\\
&=\frac{\left(p_{f'(b)}\right)^{\mu'_{b}}}{\left(1-p_{f'\left(b\right)}\right)^{\mu'_b-1}}
\prod_{j\in[k]\backslash\{b\}}\left(\frac{p_{f'\left(j\right)}}{1-p_{f'\left(j\right)}}\right)^{\mu'_{j}}\cdot
\end{align*}
By summing both sides of the equality over all bijections between the set $\left\{ 1,2,\cdots,k\right\} $
and a subset $A\subseteq\mathcal{A}$ of size $k$, it follows that $p\left( \overline\psi\right) = p\left( \overline\psi'\right)$.
\end{proof}

\section{Block Estimators for Alternating Sequences}

\subsection{Upper Bound on Redundancy of Alternating Patterns\label{sec:UB}}
We start by deriving an upper bound
on the \emph{worst case} redundancy of alternating sequences, defined, with respect to a collection
of distributions $\mathcal{P}$, as 
\begin{equation}
\begin{split}\hat{R}(\mathcal{P})=\inf_{{q}}\sup_{p\in\mathcal{P}}\sup_{u\in\mathcal{U}(p)}\log\frac{p(u)}{{q}(u)},\end{split}
\label{eq:rdndncy-def}
\end{equation}
 where $\mathcal{U}(p)$ is the support set of the distribution $p$,
${q}(u)$ is the probability assigned to $u$ by the estimator ${q}$,
and $p(u)$ is the probability of $u$ with respect to the distribution
$p$.
As already pointed out, alternating sequences are first-order Markov processes. Prior results by Dhulipala and Orlitsky \cite{Dhulipala2006} showed that in general, 
the per-symbol pattern redundancy of a first-order Markov process may be unbounded. For the particular case of alternating sequences, however, we show that the 
per-symbol pattern redundancy tends to zero.

Suppose that $\mathcal{P}$ is a collection of distributions over patterns of length $n$ and let $\Psi(\mathcal{P})$ denote the set of patterns with positive probability with respect to some distribution in $\mathcal{P}$. Denote the set of profiles of patterns in $\Psi(\mathcal P)$ by $\Phi\left(\mathcal P\right)$ and let $\tilde\Psi^n:=\Psi(\mathcal{V}_{\Psi}^{n})$ and $\tilde{\Phi}^{n}:=\Phi(\mathcal{V}_{\Psi}^{n})$. Also, for a pattern $\psi^n$ of length $n$, let the set of all alternating patterns with same profile as that of $\psi^n$ be denoted by $\tilde{\Psi}^{-1}\left(\Phi\left(\psi^n\right)\right)$.

 If $\Psi(\mathcal{P})$ is partitioned into $M $ classes such that any distribution in $\mathcal{P}$ assigns the \emph{same probability to all patterns in the same class}, then the worst case redundancy is bounded by \cite{orlitsky_universal_2004}
\begin{equation}
\hat{R}\bigl(\mathcal{P}\bigr)\le\log M.
\label{eq:rdndncy-UB}
\end{equation}



For the collection $\mathcal{P}_{\Psi}^{n}$, patterns with the same profile have the same probability and thus $M=|\Phi(\mathcal{P}_{\Psi}^{n})|$. The cardinality of $\Phi(\mathcal{P}_{\Psi}^{n})$ is equal to the number of partitions of $n$ because there is a one-to-one correspondence between distinct profiles $\Phi(x^{n})$ of \iid sequences $x^{n}$ and partitions of $n$. Not every partition corresponds to a profile of an alternating sequence. For example, for $\varphi=(0,0\upto0,1)$, a partition with one part of size $n$, there is no alternating sequence $v^{n}$ such that $\Phi(v^{n})=\varphi$ and thus $\varphi\notin\big|\tilde{\Phi}^{n}\big|$.
It is, however, easy to see that every profile of an alternating sequence corresponds to a unique partition and hence $\big|\tilde{\Phi}^{n}\big|\le\ptns n$.

\begin{thm}\label{thm:1161}
The worst case redundancy of $\mathcal{V}_{\Psi}^{n}$
grows at most linearly with $\sqrt{n}$. More precisely,
\[
\begin{split}\hat{R}(\mathcal{V}_{\Psi}^{n}) & \le\left(\pi\sqrt{\frac{2}{3}}\log e\right)\sqrt{n}+\log n\end{split}
\]
\end{thm}
\begin{proof}
From Lemma \ref{lem:same-multip}, all patterns with the same profile and the same multiplicity of the last element, have the same probability. Hence, for any distribution $p$ and any pattern $ \overline\psi\in\tilde\Psi^n$,
\begin{equation}\label{eq:maxprb}
p(\overline\psi)\le \frac 1 {L(\overline\psi)}
\end{equation}
where $L( \overline\psi)$ denote the number of patterns with the same profile and the same multiplicity of the last element as $ \overline\psi$.

Consider an estimator $q$ that assigns probability
\begin{equation}\label{eq:qL}
q\left(\overline{\psi}\right)=\frac{\nicefrac{1}{L(\overline\psi)}} {\sum_{\overline{\psi}'\in\tilde{\Psi}^{n}}\nicefrac{1}{L( \overline\psi')}}
\end{equation}
to $ \overline\psi\in\tilde\Psi^n$. We have that 
\begin{align}
\sup_{p}\sup_{\overline{\psi}\in\tilde{\Psi}^{n}}\frac{p\left(\overline{\psi}\right)}{q\left(\overline{\psi}\right)}
&\le {\sum_{\overline{\psi}'\in\tilde{\Psi}^{n}}\nicefrac{1}{L( \overline\psi')}}\\
&\le \sum_{ \overline \varphi\in \tilde{\Phi}^{n} } \sum_{k\in[n]} 
\sum_{\stackrel{\overline\psi': \Phi(\overline\psi')= \varphi,}{\mu_{\overline\psi'}(\psi'_n)=k}}
\nicefrac{1}{L( \overline\psi')}.
\end{align}
In the triple summation above, the index $k$ corresponds to the multiplicity of the last element $\psi'_n$ of $\overline\psi'$. By definition of $L(\overline\psi')$, we have
\begin{align}
\sum_{\stackrel{\overline\psi': \Phi(\overline\psi')=\overline\varphi,}{\mu_{\overline\psi'}(\psi'_n)=k}}
\nicefrac{1}{L( \overline\psi')}\le 1
\end{align}
and thus
\begin{align}
\sup_{p}\sup_{\overline{\psi}\in\tilde{\Psi}^{n}}\frac{p\left(\overline{\psi}\right)}{q\left(\overline{\psi}\right)}
&\le \sum_{ \overline \varphi\in \tilde{\Phi}^{n} } \sum_{k\in[n]} 1 \le n\big|\tilde{\Phi}^{n}\big|
\end{align}
which implies that 
\begin{equation}\label{eq:profile-UB}
\hat{R}\left(\mathcal{V}_{\Psi}^{n}\right)
\le\log\big|\tilde{\Phi}^{n}\big|+\log n\le \log \ptns n+\log n.
\end{equation}
The theorem then follows from $ \ptns n\le e^{\pi{\left(\frac{2}{3}\right)} ^{\frac{1}{2}}n^{\frac{1}{2}}}$ \cite[pp.  8--102]{andrews_theory_1998}.
\end{proof}

In \eqref{eq:profile-UB}, $\ptns n$ is used as an upper bound for $\big|\tilde{\Phi}^{n}\big|$. It may seem possible that a tighter upper bound for $\big|\tilde{\Phi}^{n}\big|$ improves Theorem \ref{thm:1161}. However, the following lemma shows that this is not the case and $\ptns n $ is sufficiently tight.
\begin{lem}
\label{lem:card-alt-profiles}For the cardinality of $\tilde{\Phi}^{n}$
we have
\[
\big|\tilde{\Phi}^{n}\big|=\ptns n\left(1-\left|O(\sqrt{n}e^{-\sqrt{n}})\right|\right)\cdot
\]
\end{lem}
\begin{rem}
The relevant problem of finding the number of partitions of $n$ into
at most $k$ parts for $k\ge n^{-1/6}$ was studied by Szekeres \cite{szekeres_asymptotic_1951,szekeres_asymptotic_1953}.
Our proof however is much simpler since it considers a special case
where $k\ge n/2$.
\end{rem}
\begin{proof}
We first show that there is a one-to-one correspondence between $\tilde{\Phi}^{n}$
and partitions of $n$ with no part larger than $\left\lfloor \frac{n+1}{2}\right\rfloor $.
Clearly, each $\overline{\varphi}\in\tilde{\Phi}^{n}$ determines
a unique partition of $n$. For a profile $\varphi^{n}$ with a part
$\mu$ with size larger than $\left\lfloor \frac{n+1}{2}\right\rfloor $
suppose $x^{n}$ is some sequence such that $\varphi^{n}=\Phi\left(x^{n}\right)$.
There is a symbol in $x^{n}$ that appears $\mu$ times, say $a$.
We need at least $\mu-1$ other symbols to separate every two occurrences
of $a$. However, this is not possible since $\mu+\mu-1>n$ and thus
$x^{n}$ is not an alternating sequence. On the other hand, if all
parts are of size at most $\left\lfloor \frac{n+1}{2}\right\rfloor $,
occurrences of every symbol can be separated by other symbols. This
bijection implies that  
\begin{equation}
\big|\tilde{\Phi}^{n}\big|=\ptns {n,\left\lfloor \frac{n+1}{2}\right\rfloor}\label{eqn:card-alt-profiles}
\end{equation}
where $\ptns {n,r}$ denote the number of partitions of $n$ with largest part of size at most $r$.
Furthermore, 
\begin{align*}
\ptns {n,\left\lfloor \frac{n+1}{2}\right\rfloor}  =\ptns n\biggl(1-\sum_{i=\left\lfloor \frac{n+1}{2}\right\rfloor +1}^{n}\frac{\ptns{n-i}}{\ptns n}\biggr).
\end{align*}
 The lemma then follows by using \cite{odlyzko_asymptotic_1995},
$\ptns{n-i}/\ptns n=\left(1+O\left(n^{-1/6}\right)\right)e^{\frac{-i\pi}{\sqrt{6n}}}.$
\end{proof}

\subsection{Lower Bound on the Redundancy of Alternating Patterns\label{sec:lower-bound}}

In subsection \ref{sec:UB}, we saw that the redundancy of patterns of
alternating sequences is $O\left(\sqrt{n}\right)$. Here, we show
that it is bounded from below by a constant multiple of $n^{1/3}$.
\begin{lem}\label{lem:Prob-LB}
Let $\overline{\psi}=1\psi_{1}1\psi_{2}\cdots1\psi_{n/2}$, for even
$n$, and $\overline{\psi}=1\psi_{1}1\psi_{2}\cdots1\psi_{\left\lfloor n/2\right\rfloor }1$, for odd $n$, be an alternating pattern. For a function $r_{n}\ge1$
of $n$, we have 
\begin{equation}
\sup_{p\in\mathcal{V}_{\Psi}^{n}}p\left(\overline{\psi}\right)\label{eqn:maxprob-tilde}\ge\frac{1}{2}\left(\frac{2r_{n}-1}{2r_{n}}\right)^{\left\lfloor n/2\right\rfloor }\prod_{\mu=1}^{\left\lfloor n/2\right\rfloor }\varphi_{\mu}!\left(\frac{\mu}{\left\lfloor n/2\right\rfloor }\right)^{\mu\varphi_{\mu}},
\end{equation}
where $\overline{\varphi}$ is the profile of the pattern $\psi_{1}\psi_{2}\cdots\psi_{\left\lfloor n/2\right\rfloor }$. \end{lem}
\begin{proof}
Note that since $\overline{\psi}$ is an alternating pattern, we have
$\psi_{j}\neq1$ for $1\le j\le\left\lfloor n/2\right\rfloor $. Consider
the alphabet $\mathcal{A}=\left\{ a,s_{1},\cdots,s_{m}\right\},$
where $m$ is the largest number appearing in $\overline{\psi}$ minus
one. Let $\overline{v}$ be a sequence with pattern $\overline{\psi}$
starting with symbol $a$. Let $p$ be a distribution defined as 
\[
p_{i}=\begin{cases}
\frac{\mu\left(i,\overline{v}\right)}{nr_{n}}, & \qquad i\neq a,\\
1-\frac{1}{2r_{n}}, & \qquad i=a,
\end{cases}
\]
where $\mu\left(i,\overline{v}\right)$ is the number of occurrences
of $i$ in $\overline{v}$. First, suppose that $n$ is even. We then have
\begin{align*}
&p\left(\overline{v}\right) =\frac{p_{a}p_{v_{1}}}{1-p_{a}}\prod_{j=2}^{n/2}\left(\frac{p_{a}}{1-p_{v_{j-1}}}\frac{p_{v_{j}}}{1-p_{a}}\right)\\
 & \ge\left(\frac{p_{a}}{1-p_{a}}\right)^{n/2}\prod_{j=1}^{n/2}p_{v_{j}}
 =\left(\frac{2r_{n}-1}{2r_{n}}\right)^{n/2}\prod_{\mu=1}^{n/2}\left(\frac{\mu}{n/2}\right)^{\mu\varphi_{\mu}}.
\end{align*}
The position of $a$ is fixed in $\overline{v}$, but the symbols $\left\{ s_{1},\cdots,s_{m}\right\} $
that appear the same number of times can be swapped without changing
the pattern $\Psi\left(\overline{v}\right)$ of $\overline{v}$. This
can be done in $\prod_{\mu=1}^{n/2-1}\varphi_{\mu}!$ ways. Hence,
there are $\prod_{\mu=1}^{n/2-1}\varphi_{\mu}!$ sequences with pattern
$\overline{\psi}$ and probability $p\left(\overline{v}\right)$.
Since $\varphi_{n/2}\le2$, we have 
\[
\prod_{\mu=1}^{n/2-1}\varphi_{\mu}!\ge\frac{1}{2}\prod_{\mu=1}^{n/2}\varphi_{\mu}!.
\]
 Thus, 
\begin{align*}
\sup_{p\in\mathcal{V}_{\Psi}^{n}}p\left(\overline{\psi}\right)&\ge\frac{1}{2}\left(\prod_{\mu=1}^{n/2}\varphi_{\mu}!\right)p\left(\overline{v}\right)\\
&=\frac{1}{2}\left(\frac{2r_{n}-1}{2r_{n}}\right)^{n/2}\prod_{\mu=1}^{n/2}\varphi_{\mu}!\left(\frac{\mu}{n/2}\right)^{\mu\varphi_{\mu}}\cdot
\end{align*}

Next suppose $n$ is odd. Then,
\begin{align*}
p\left(\overline{v}\right) & =p_{a}\prod_{j=1}^{\left\lfloor n/2\right\rfloor }\left(\frac{p_{v_{j}}}{1-p_{a}}\frac{p_{a}}{1-p_{v_{j-1}}}\right)\\
&\ge\frac{1}{2}\left(\frac{p_{a}}{1-p_{a}}\right)^{\left\lfloor n/2\right\rfloor }\prod_{j=1}^{\left\lfloor n/2\right\rfloor }p_{v_{j}}\\
 & =\frac{1}{2}\left(\frac{2r_{n}-1}{2r_{n}}\right)^{\left\lfloor n/2\right\rfloor }\prod_{\mu=1}^{\left\lfloor n/2\right\rfloor }\left(\frac{\mu}{\left\lfloor n/2\right\rfloor }\right)^{\mu\varphi_{\mu}}
\end{align*}
where the inequality follows since $p_{a}\ge1/2$. The number of sequences
with pattern $\overline{\psi}$ and starting with $a$ is $\prod_{\mu=1}^{\left\lfloor n/2\right\rfloor }\varphi_{\mu}!$.\end{proof}

\begin{thm}\label{thm:LB}
For the collection $\mathcal{V}_{\Psi}^{n}$ of distributions,
\[
\hat{R}\left(\mathcal{V}_{\Psi}^{n}\right)\ge\frac{1}{2^{1/3}}\log\left(\frac{e^{23/12}}{\sqrt{2\pi}}\right)n^{1/3}\left(1+o\left(1\right)\right).
\]
\end{thm}
\begin{proof}
From Starkov's sum~\cite{orlitsky_universal_2004} we have that
\[
\hat{R}\left(\mathcal{V}_{\Psi}^{n}\right)=\log\left(\sum_{\overline{\varphi}\in\Phi\left(\mathcal{V}_{\Psi}^{n}\right)}\sum_{\overline{\psi}\in\Psi_{\overline{\varphi}}}\sup_{p\in\mathcal{V}_{\Psi}^{n}}p\left(\overline{\psi}\right)\right).
\]
where $\Psi_{\overline{\varphi}}$ is the set of all patterns with
profile $\overline{\varphi}$.

Suppose first that $n$ is even. Let $\Phi_{k}^{n}$ be the set of
profiles whose largest parts are of size $k$. Since $\Phi_{n/2}^{n}\subseteq\Phi\left(\mathcal{V}_{\Psi}^{n}\right)$,
\begin{align*}
&\hat{R}\left(\mathcal{V}_{\Psi}^{n}\right)  \ge\log\left(\sum_{\overline{\varphi}\in\Phi_{n/2}^{n}}\quad\sum_{\overline{\psi}\in\Psi_{\varphi}}\quad\sup_{p\in\mathcal{V}_{\Psi}^{n}}p\left(\overline{\psi}\right)\right)\\
 & \ge\log\left(\sum_{\overline{\varphi}\in\Phi_{n/2}^{n}}\quad\sum_{\overline{\psi}\in\Psi_{\varphi},\mu\left(1,\overline{\psi}\right)=n/2}\quad\sup_{p\in\mathcal{V}_{\Psi}^{n}}p\left(\overline{\psi}\right)\right)\\
 & =\log\left(\sum_{\overline{\varphi}\in\Phi^{n/2}}\quad\sum_{\overline{\psi}\in\Psi_{\overline{\varphi}}}\quad\sup_{p\in\mathcal{V}_{\Psi}^{n}}p\left(\Psi_{A}\left(\overline{\psi}\right)\right)\right),
\end{align*}
where $\Psi_{A}\left(\overline{\psi}\right)$ is the alternating pattern
$\left(1,\psi_{1}+1,1,\psi_{2}+1,\cdots,1,\psi_{n/2}+1\right)$ obtained
from $\overline{\psi}$. 

From \eqref{eqn:maxprob-tilde}, we obtain \eqref{eq:long} in which where (a) and (b) follow from Lemma 3 in \cite{orlitsky_universal_2004} and the proof of Theorem 13 in \cite{orlitsky_universal_2004}, respectively.
\begin{figure*}
\begin{align}\label{eq:long}
\hat{R}\left(\mathcal{V}_{\Psi}^{n}\right) & \ge\log\left(\sum_{\overline{\varphi}\in\Phi^{n/2}}\sum_{\overline{\psi}\in\Psi_{\overline{\varphi}}}\left(\frac{2r_{n}-1}{2r_{n}}\right)^{n/2}\prod_{\mu=1}^{n/2}\varphi_{\mu}!\left(\frac{\mu}{n/2}\right)^{\mu\varphi_{\mu}}\right)\\
 & \stackrel{\left(\text{a}\right)}{\ge}\log\left(\sum_{\overline{\varphi}\in\Phi^{n/2}}\frac{\left(n/2\right)!}{\prod_{\mu=1}^{n/2}\left(\mu!\right)^{\varphi_{\mu}}\varphi_{\mu}!}\left(\frac{2r_{n}-1}{2r_{n}}\right)^{n/2}\prod_{\mu=1}^{n/2}\varphi_{\mu}!\left(\frac{\mu}{n/2}\right)^{\mu\varphi_{\mu}}\right)\nonumber\\
 & \stackrel{}{\ge}\frac{n}{2}\log\left(1-\frac{1}{2r_{n}}\right)+\log\left(\sum_{\overline{\varphi}\in\Phi^{n/2}}\frac{\left(n/2\right)!}{\prod_{\mu=1}^{n/2}\left(\mu!\right)^{\varphi_{\mu}}\varphi_{\mu}!}\prod_{\mu=1}^{n/2}\varphi_{\mu}!\left(\frac{\mu}{n/2}\right)^{\mu\varphi_{\mu}}\right)\nonumber\\
 & \stackrel{\left(\text{b}\right)}{\ge}\frac{n}{2}\log\left(1-\frac{1}{2r_{n}}\right)+\log\left(\frac{e^{23/12}}{\sqrt{2\pi}}\right)\left(n/2\right)^{1/3}\left(1+o\left(1\right)\right)\nonumber\\
 & \stackrel{}{\ge}\frac{n}{2}\cdot\frac{1}{r_{n}}+\frac{1}{2^{1/3}}\log\left(\frac{e^{23/12}}{\sqrt{2\pi}}\right)n^{1/3}\left(1+o\left(1\right)\right)\nonumber\\
 & \stackrel{}{\ge}\frac{1}{2^{1/3}}\log\left(\frac{e^{23/12}}{\sqrt{2\pi}}\right)n^{1/3}\left(1+o\left(1\right)\right)\nonumber
\end{align}
\rule{\linewidth}{0.2mm}
\end{figure*}
The proof for odd $n$ is similar.
\end{proof}
\section{Sequential Estimators for Alternating Sequences \label{sec:Sequential-Estimators}}

In the previous section we studied the problem of assigning probabilities to patterns of a certain length without prior information. 
In this section, we address a more practical problem. Namely, given a pattern $\psi^{n-1}$ of length $n-1$, what is our best estimate $q\left(\psi_{n}|\psi^{n-1}\right)$ of the probability of $\psi_{n}$ being the next observed symbol, for $\psi_{n}\in\left\{ 1,2,\cdots,1+\max_{i\le n-1}\psi_{i}\right\} $. This sequential estimator also assigns probabilities to patterns of length $n$ in a natural way. That is, 
\[
q\left(\psi^n\right)=\prod_{i=1}^{n}q\left(\psi_{i}|\psi^{i-1}\right),
\]
where $\psi^{0}$ is an empty string.

We present a sequential estimator $q_{1/2}$ for patterns of alternating
sequences which is based on a sequential estimator for patterns of
\iid sequences presented in \cite{orlitsky_universal_2004} by Orlitsky
et al. 

Let 
\[
\hat{p}_{\psi^n}\left(\psi^n\right):\,=\sup_{p\in\mathcal{V}_{\Psi}^{n}}p\left(\psi^n\right)
\]
be the largest probability assigned to $\psi^n$ by any distribution
in $\mathcal{V}_{\Psi}^{n}$ and let $q$ be as defined in \eqref{eq:qL}, i.e.,
\begin{equation}\label{eq:qLL}
q\left({\psi}^n\right)=\frac{\nicefrac{1}{L(\psi^n)}} {\sum_{\overline{\psi}\in\tilde{\Psi}^{n}}\nicefrac{1}{L( \overline\psi)}},
\end{equation}
for which we have
\begin{equation}
\frac{\hat{p}_{\psi^n}\left(\psi^n\right)}{q \left(\psi^n\right)}
\le n\exp\left(\pi\sqrt{\frac{2}{3}}\sqrt{n}\right) \label{eqn:phatn}
\end{equation}

For an alternating pattern $\psi^{i}$, let \[\tilde{\Psi}^{n}\left(\psi^{i}\right):=\left\{ \overline{z}\in\tilde{\Psi}^{n}:z_{1}z_{2}\cdots z_{i}=\psi^{i}\right\} \]
be the set of alternating patterns of length $n\ge i$ whose first
$i$ elements are the same as $\psi^{i}$. Accordingly, from $ q \left(\overline{z}\right),\overline{z}\in\tilde{\Psi}^{n}$,
we define the distribution 
\[
 q ^{n}\left(\psi^{i}\right):=\sum_{\overline{z}\in\tilde{\Psi}^{n}\left(\psi^{i}\right)} q \left(\overline{z}\right)
\]
over $\tilde{\Psi}^{i}$ for $i\le n$.

We define the estimator $\qhalfn$ such that $\qhalfn\left(1\right):=1$
and 
\begin{align*}
\qhalfn\left(\psi_{i}|\psi^{i-1}\right) & :=\frac{ q ^{n}\left(\psi^{i}\right)}{ q ^{n}\left(\psi^{i-1}\right)},
\end{align*}
for $i\le n$. Note that $\qhalfn\left(\psi^n\right)= q ^{n}\left(\psi^n\right)$
and thus, by \eqref{eqn:phatn}, we have the following theorem.
\begin{thm}
The redundancy of $\qhalfn$ at time $n$ is sub-linear in $n$. Namely,
\begin{align*}
\hat{R}\left(\qhalfn,\mathcal{V}_{\Psi}^{n}\right)&=\sup_{\psi^n\in\tilde{\Psi}^{n}}\log\frac{\hat{p}_{\psi^n}\left(\psi^n\right)}{\qhalfn\left(\psi^n\right)}\\&\le\left(\pi\sqrt{\frac{2}{3}}\log e\right) \sqrt{n}+\log n\cdot
\end{align*}

\end{thm}
Note that $\qhalfn$ has the drawback that it is applicable only to
patterns with predetermined length $n$. Such estimators are called
horizon-dependent \cite{cesa2006prediction}; assigned probabilities
depend on the ``horizon'' $n$. Thus $\qhalfn$ cannot be used to
sequentially estimate the probabilities for a pattern whose length
is unknown. 

However, it is easy to remove this restriction using the so-called
``doubling trick'' where time is divided into periods each twice
as long as its predecessors. That is, the horizon $h_{i}$ at time
$i$ is considered to be $2^{\left\lceil \log i\right\rceil }$, the
smallest power of two which is at least as large $i$. The estimator
$\qhalf$, where 
\begin{gather*}
\qhalf\left(1\right)=1,\qquad\qhalf\left(\psi^{i}|\psi_{i-1}\right)=\frac{ q ^{h_{i}}\left(\psi^{i}\right)}{ q ^{h_{i}}\left(\psi^{i-1}\right)},
\end{gather*}
is thus horizon-independent and the following theorem holds uniformly
over time.

\begin{thm}
The worst case redundancy of the sequential estimator $q_{1/2}$ is
bounded by 
\[
\hat{R}\left(\mathcal{V}_{\Psi}^{n},q_{1/2}\right)\le2+\frac{5}{2}\log n+\frac{1}{2}\log^{2}n+\frac{4\pi\log e\sqrt{n}}{\sqrt{3}\left(2-\sqrt{2}\right)}.
\]
\end{thm}
\begin{proof}
By definition, 
\[
\hat{R}\left(\mathcal{V}_{\Psi}^{n},q_{1/2}\right)\le\max_{\psi_{1}^{n}\in\tilde{\Psi}^{n}}\log\frac{\hat{p}_{\psi_{1}^{n}}\left(\psi_{1}^{n}\right)}{q_{1/2}\left(\psi_{1}^{n}\right)}\cdot
\]
Write 
\[
\frac{\hat{p}_{\psi_{1}^{n}}\left(\psi_{1}^{n}\right)}{q_{1/2}\left(\psi_{1}^{n}\right)}=\frac{\hat{p}_{\psi_{1}^{n}}\left(\psi_{1}^{n}\right)}{q^{h_{n}}\left(\psi_{1}^{n}\right)}\cdot\frac{q^{h_{n}}\left(\psi_{1}^{n}\right)}{q_{1/2}\left(\psi_{1}^{n}\right)}\cdot
\]
From Lemmas \ref{lem:16} and \ref{lem:sequential}, we obtain
\[
\frac{\hat{p}_{\psi_{1}^{n}}\left(\psi_{1}^{n}\right)}{q_{1/2}\left(\psi_{1}^{n}\right)}\le h_{n}^{1+\left(\log h_{n}-1\right)/2}\exp\left(\pi\sqrt{\frac{2}{3}}\frac{2\sqrt{h_{n}}}{2-\sqrt{2}}\right)\cdot
\]
The theorem follows after some minor algebra and by noting that $h_{n}<2n$.\end{proof}
\begin{lem}
\label{lem:16}For an alternating pattern $\psi_{1}^{n}$,
we have 
\[
\frac{\hat{p}_{\psi_{1}^{n}}\left(\psi_{1}^{n}\right)}{q^{h_{n}}\left(\psi_{1}^{n}\right)}\le h_{n}\exp\left(\pi\sqrt{\frac{2}{3}}\sqrt{h_{n}}\right)\cdot
\]
\end{lem}
\begin{proof}
For any \iid induced distribution $p$ over alternating patterns, and
for $t\ge n$, note that 
\[
\sum_{\overline{z}\in\tilde{\Psi}^{t}\left(\psi_{1}^{n}\right)}p\left(\overline{z}\right)=p\left(\psi_{1}^{n}\right).
\]
Hence, 
\begin{align}
\hat{p}_{\psi_{1}^{n}}\left(\psi_{1}^{n}\right) & =\sup_{p}\sum_{\overline{z}\in\tilde{\Psi}^{h_{n}}\left(\psi_{1}^{n}\right)}p\left(\overline{z}\right)\nonumber \\
 & \le\sum_{\overline{z}\in\tilde{\Psi}^{h_{n}}\left(\psi_{1}^{n}\right)}\hat{p}_{\overline{z}}\left(\overline{z}\right).\label{eq:q1}
\end{align}
Using \eqref{eqn:phatn}, we can write 
\begin{align}
&\sum_{\overline{z}\in\tilde{\Psi}^{h_{n}}\left(\psi_{1}^{n}\right)}\hat{p}_{\overline{z}}\left(\overline{z}\right)\\ & \le h_{n}\exp\left(\pi\sqrt{\frac{2}{3}}\sqrt{h_{n}}\right)\sum_{\overline{z}\in\tilde{\Psi}^{h_{n}}\left(\psi_{1}^{n}\right)}q\left(\overline{z}\right).\label{eq:q2}
\end{align}
Furthermore, observe that 
\begin{equation}
\sum_{\overline{z}\in\tilde{\Psi}^{h_{n}}\left(\psi_{1}^{n}\right)}q\left(\overline{z}\right)=q^{h_{n}}\left(\psi_{1}^{n}\right).\label{eq:q3}
\end{equation}
From \eqref{eq:q1}, \eqref{eq:q2}, and \eqref{eq:q3}, we obtain
the desired result.\end{proof}

\begin{lem}
\label{lem:sequential}For $n\ge2$ and an alternating pattern $\psi_{1}^{n}$,
we have
\[
\frac{q^{h_{n}}\left(\psi_{1}^{n}\right)}{q_{1/2}\left(\psi_{1}^{n}\right)}\le h_{n}^{\left(\log h_{n}-1\right)/2}\exp\left(\pi\sqrt{\frac{2}{3}}\frac{\sqrt{h_{n}}}{\sqrt{2}-1}\right).
\]
\end{lem}
\begin{proof}
We show inductively that 
\[
\frac{q^{2^{i+1}}\left(\psi_{1}^{n}\right)}{q_{1/2}\left(\psi_{1}^{n}\right)}\le\left(2^{i+1}\right)^{i/2}\exp\left(\pi\sqrt{\frac{2}{3}}\frac{\sqrt{2^{i+1}}}{\sqrt{2}-1}\right)
\]
for $2^{i}<n\le2^{i+1}$ and an alternating pattern $\psi_{1}^{n}$.
This shall prove the lemma since $h_{n}=2^{i+1}$. 

From the definiton of $q_{1/2}$, it follows that 
\begin{equation}
\frac{q^{2^{i+1}}\left(\psi_{1}^{n}\right)}{q_{1/2}\left(\psi_{1}^{n}\right)}=\frac{q^{2^{i+1}}\left(\psi_{1}^{2^{i}}\right)}{q_{1/2}\left(\psi_{1}^{2^{i}}\right)}=\frac{q^{2^{i+1}}\left(\psi_{1}^{2^{i}}\right)}{q^{2^{i}}\left(\psi_{1}^{2^{i}}\right)}\cdot\frac{q^{2^{i}}\left(\psi_{1}^{2^{i}}\right)}{q_{1/2}\left(\psi_{1}^{2^{i}}\right)}\cdot\label{eq:q-chain}
\end{equation}

As the induction hypothesis, we have that 
\begin{equation}
\frac{q^{2^{i}}\left(\psi_{1}^{2^{i}}\right)}{q_{1/2}\left(\psi_{1}^{2^{i}}\right)}\le2^{i\left(i-1\right)/2}\exp\left(\pi\sqrt{\frac{2}{3}}\frac{\sqrt{2^{i}}}{\sqrt{2}-1}\right).\label{eq:induction}
\end{equation}

All patterns in $L\left(\psi_{1}^{2^{i}}\right)$ have the same assigned
probability $q^{2^{i+1}}\left(\psi_{1}^{2^{i}}\right)$. Hence,
\[
\sum_{\overline{\psi}\in L\left(\psi_{1}^{2^{i}}\right)}q^{2^{i+1}}\left(\overline{\psi}\right)=\left|L\left(\psi_{1}^{2^{i}}\right)\right|q^{2^{i+1}}\left(\psi_{1}^{2^{i}}\right)\cdot
\]
On the other hand, for all patterns $\overline{\psi}\in L\left(\psi_{1}^{2^{i}}\right)$
the sets $\tilde{\Psi}^{2^{i+1}}\left(\overline{\psi}\right)$ are
disjoint. This implies that 
\begin{align*}
\sum_{\overline{\psi}\in L\left(\psi_{1}^{2^{i}}\right)}q^{2^{i+1}}\left(\overline{\psi}\right)&=\sum_{\overline{\psi}\in L\left(\psi_{1}^{2^{i}}\right)}\sum_{\overline{z}\in\tilde{\Psi}^{2^{i+1}}\left(\overline{\psi}\right)}q\left(\overline{z}\right)\\
&\le\sum_{\overline{z}\in\tilde{\Psi}^{2^{i+1}}}q\left(\overline{z}\right)\le1.
\end{align*}
Thus, we obtain
\begin{equation}
q^{2^{i+1}}\left(\psi_{1}^{2^{i}}\right)\le\frac{1}{\left|L\left(\psi_{1}^{2^{i}}\right)\right|}\label{eq:q-up}
\end{equation}
From \eqref{eq:qLL}, we have 
\begin{equation}
q^{2^{i}}\left(\psi_{1}^{2^{i}}\right)\ge\frac{1}{2^{i}\exp\left(\pi\sqrt{\frac{2}{3}}\sqrt{2^{i}}\right)\left|L\left(\psi_{1}^{2^{i}}\right)\right|}\cdot\label{eq:q-low}
\end{equation}
From \eqref{eq:q-up} and \eqref{eq:q-low}, we find
\[
\frac{q^{2^{i+1}}\left(\psi_{1}^{2^{i}}\right)}{q^{2^{i}}\left(\psi_{1}^{2^{i}}\right)}\le2^{i}\exp\left(\pi\sqrt{\frac{2}{3}}\sqrt{2^{i}}\right)\cdot
\]
This inequality, along with \eqref{eq:q-chain} and \eqref{eq:induction},
complete the proof.
\end{proof}

\section{Estimating Distribution of Source}\label{sec:recovery}
In this section, we explain how to reconstruct the noiseless source probabilities from estimates of probabilities provided by alternating sequences. 
First, recall from Lemma \ref{lem:n-v-N}, that with high probability, $n$ is of the same order as $N$ and thus, with high probability, the length $n$ of the alternating sequence is large if the length of the source sequence is large.

 Assume that the source has alphabet $\mathcal{A}=\{a_1,a_2,\cdots\}$ with probability $p_{a_j}$ for element $a_j$. Suppose $p_{a_ia_j}$ is the probability of observing $a_j$ after $a_i$ in the alternating sequence and assume that the correct values of $p_{a_1a_2}$ and $p_{a_2a_1}$ are given. We have
\begin{equation*}
p_{a_1a_2} = \frac{p_{a_2}}{1-p_{a_1}},\qquad
p_{a_2a_1} = \frac{p_{a_1}}{1-p_{a_2}}
\end{equation*}
which implies that $p_{a_1}$ and $p_{a_2}$ can be found by
\begin{align*}
p_{a_1} &= p_{a_2a_1}\frac{1-p_{a_1a_2}}{1-p_{a_1a_2}p_{a_2a_1}},  \\ \nonumber
p_{a_2} &= p_{a_1a_2}\frac{1-p_{a_2a_1}}{1-p_{a_1a_2}p_{a_2a_1}}\cdot
\end{align*}

Given $p_{a_1a_j}$ for $j\ge2$, the remaining probabilities may be obtained by noting that $\frac{p_{a_1a_j}}{p_{a_1a_2}}=\frac{p_{a_j}}{p_{a_2}}$ and thus
\[p_{a_j}=p_{a_2} \frac{p_{a_1a_j}}{p_{a_1a_2}}
\]
gives the probabilities $p_{a_j}$ for $j\ge3$.

Although as with any estimator, the estimators presented here for the alternating sequence do not find probabilities with zero error, we are justified in assuming that the estimates obtained from these estimators are ``close'' to the correct values since their redundancy is vanishing. Hence, the estimates of the probabilities $p_{a_ia_j}$ of the alternating sequence  can be used to obtain estimates for probabilities $p_i$ of the source sequence as explained above.
\bibliographystyle{IEEEtran}
\bibliography{bib}

\begin{thebibliography}{10}
\providecommand{\url}[1]{#1}
\csname url@samestyle\endcsname
\providecommand{\newblock}{\relax}
\providecommand{\bibinfo}[2]{#2}
\providecommand{\BIBentrySTDinterwordspacing}{\spaceskip=0pt\relax}
\providecommand{\BIBentryALTinterwordstretchfactor}{4}
\providecommand{\BIBentryALTinterwordspacing}{\spaceskip=\fontdimen2\font plus
\BIBentryALTinterwordstretchfactor\fontdimen3\font minus
  \fontdimen4\font\relax}
\providecommand{\BIBforeignlanguage}[2]{{%
\expandafter\ifx\csname l@#1\endcsname\relax
\typeout{** WARNING: IEEEtran.bst: No hyphenation pattern has been}%
\typeout{** loaded for the language `#1'. Using the pattern for}%
\typeout{** the default language instead.}%
\else
\language=\csname l@#1\endcsname
\fi
#2}}
\providecommand{\BIBdecl}{\relax}
\BIBdecl

\bibitem{good_population_1953}
I.~J. Good, ``The population frequencies of species and the estimation of
  population parameters,'' \emph{Biometrika}, vol.~40, no. 3-4, pp. 237--264,
  1953.

\bibitem{gale_good-turing_1995}
W.~A. Gale and G.~Sampson, ``{Good-Turing} smoothing without tears,''
  \emph{Journal of Quantitative Linguistics}, vol.~2, 1995.

\bibitem{orlitsky_always_2003}
A.~Orlitsky, N.~Santhanam, and J.~Zhang, ``Always good turing: asymptotically
  optimal probability estimation,'' in \emph{Foundations of Computer Science,
  2003. Proceedings. 44th Annual {IEEE} Symposium on}, 2003, pp. 179--188.

\bibitem{orlitsky-convergence2005}
A.~Orlitsky, N.~Santhanam, K.~Viswanathan, and J.~Zhang, ``Convergence of
  profile based estimators,'' in \emph{Proc. {IEEE} Int. Symp. Information
  Theory}, Adelaide, SA, September 2005.

\bibitem{small_sample_sticky}
F.~Farnoud, O.~Milenkovic, and N.~Santhanam, ``Small-sample distribution
  estimation over sticky channels,'' in \emph{IEEE Int. Symp. Information
  Theory}, 28 2009-july 3 2009, pp. 1125 --1129.

\bibitem{Dhulipala2006}
A.~Dhulipala and A.~Orlitsky, ``Universal compression of markov and related
  sources over arbitrary alphabets,'' \emph{Information Theory, IEEE
  Transactions on}, vol.~52, no.~9, pp. 4182 --4190, sept. 2006.

\bibitem{alon2008probabilistic}
N.~Alon and J.~Spencer, \emph{The probabilistic method}, 3rd~ed.\hskip 1em plus
  0.5em minus 0.4em\relax John Wiley and Sons, Inc., 2008.

\bibitem{orlitsky_universal_2004}
A.~Orlitsky, N.~Santhanam, and J.~Zhang, ``Universal compression of memoryless
  sources over unknown alphabets,'' \emph{Information Theory, {IEEE}
  Transactions on}, vol.~50, no.~7, pp. 1469--1481, 2004.

\bibitem{andrews_theory_1998}
G.~E. Andrews, \emph{The theory of partitions}.\hskip 1em plus 0.5em minus
  0.4em\relax Cambridge University Press, 1998.

\bibitem{szekeres_asymptotic_1951}
G.~Szekeres, ``An asymptotic formula in the theory of partitions,''
  \emph{Quarterly Journal of Mathematics}, vol.~2, no.~1, pp. 85--108, 1951.

\bibitem{szekeres_asymptotic_1953}
------, ``Some asymptotic formulae in the theory of partitions {(II)},''
  \emph{Quarterly Journal of Mathematics}, vol.~4, no.~1, pp. 96--111, 1953.

\bibitem{odlyzko_asymptotic_1995}
A.~Odlyzko and P.~Flajolet, \emph{Asymptotic Enumeration Methods}.\hskip 1em
  plus 0.5em minus 0.4em\relax Amsterdam: Elsevier, 1995, vol.~2, pp.
  1063--1229.

\bibitem{cesa2006prediction}
N.~Cesa-Bianchi and G.~Lugosi, \emph{Prediction, learning, and games}.\hskip
  1em plus 0.5em minus 0.4em\relax Cambridge Univ Pr, 2006.

\end{thebibliography}
\end{document}